\documentclass[letterpaper, 10 pt, journal]{IEEEtran}  % Comment this line out
\IEEEoverridecommandlockouts
\usepackage{braket}
\usepackage{amsmath,amssymb}  
\usepackage{amsthm}
\usepackage{graphicx}
\usepackage{lipsum}
\usepackage{algpseudocodex}
\usepackage{algorithm}
\usepackage{verbatim}
\usepackage{subcaption}
\usepackage{doi}
\usepackage{url}
\usepackage{hyperref}
\hypersetup{
    colorlinks=true,
    linkcolor=blue,
    filecolor=magenta,      
    urlcolor=cyan,
    }

\newtheorem{definition}{Definition}
\newtheorem{theorem}{Theorem}

\newtheorem{lemma}{Lemma}

\newtheorem{assumption}{Assumption}

\IEEEoverridecommandlockouts                              % This command is only
\title{\LARGE \bf
Robust feedback-based quantum optimization: analysis of coherent control errors
}

\author{Mirko Legnini and Julian Berberich % <-this % stops a space
\thanks{This work was funded by Deutsche Forschungsgemeinschaft (DFG, German Research Foundation) under Germany’s Excellence Strategy-EXC 2075-390740016.We acknowledge the support by the Stuttgart Center for Simulation Science (SimTech).}% <-this % stops a space
\thanks{Mirko Legnini,  University of Bologna, Bologna, Italy.
        }%
\thanks{Julian Berberich, University of Stuttgart, Institute for Systems Theory and Automatic Control,  70569 Stuttgart, Germany
        {\tt\small julian.berberich@ist.uni-stuttgart.de}}%
}

\begin{document}

\maketitle

%%%%%%%%%%%%%%%%%%%%%%%%%%%%%%%%%%%%%%%%%%%%%%%%%%%%%%%%%%%%%%%%%%%%%%%%%%%%%%%%
\begin{abstract}
The Feedback-based Algorithm for Quantum Optimization (FALQON) is a Lyapunov inspired quantum algorithm proposed to tackle combinatorial optimization problems. 
In this paper, we examine the robustness of FALQON against coherent control errors, a class of multiplicative errors that affect the control input. 
We show that the algorithm is asymptotically robust with respect to systematic errors, and we derive robustness bounds for independent errors.  
Finally, we propose a robust version of FALQON which minimizes a regularized Lyapunov function. 
Our theoretical results are supported through simulations.  

\end{abstract}
\IEEEpubid{\begin{minipage}{\textwidth}\ \\[30pt] \\ \\
    \copyright  2025 IEEE. Personal use of this material is permitted. Permission from IEEE must be obtained for all other uses, in any current or future media, including reprinting/republishing this material for advertising or promotional purposes, creating new collective works, for resale or redistribution to servers or lists, or reuse of any copyrighted component of this work in other works.
\end{minipage}
}

%%%%%%%%%%%%%%%%%%%%%%%%%%%%%%%%%%%%%%%%%%%%%%%%%%%%%%%%%%%%%%%%%%%%%%%%%%%%%%%%
\section{INTRODUCTION}
Quantum computing is an alternative computing technique that utilizes quantum phenomena to process information, potentially providing better solutions to problems that are considered intractable in classical computation \cite{bible, berberich2023quantumcomputinglenscontrol}.
Quantum algorithms are represented as circuits acting on information units called qubits (quantum bits). 
We are currently in what is known as the Noisy Intermediate Scale Quantum (NISQ) era \cite{Preskill_2018}. The real quantum hardware available has reached up to hundreds of qubits and circuit execution is still subject to multiple noise sources.  

A variety of algorithms that could prove a quantum advantage over classical algorithms have been proposed, most notably Shor's algorithm for factorization \cite{Shor_1997}.
In this work, we  focus on the application of quantum algorithms to optimization problems analyzed through a control theory perspective. In particular, we  address robustness problems on the execution of these algorithms on noisy hardware.
%\subsection{Related Work}
\\

The Quantum Approximate Optimization Algorithm (QAOA) \cite{farhi2014quantumapproximateoptimizationalgorithm} is an iterative algorithm proposed to solve combinatorial optimization problems with quantum systems. QAOA encodes a combinatorial optimization problem in a so-called problem Hamiltonian and uses an iterative optimization procedure to find its minimum eigenvalue, thus providing the optimal solution string. 
The main problem of this iterative approach is it can lead to complicacies due to the intrinsic structure of the optimization landscape. Barren plateaus are an important obstacle \cite{McClean2018}.  
A recent result \cite{kazi2024analyzingquantumapproximateoptimization} even suggests that, if barren plateaus are not present, the algorithm may be classically simulable. 
Possible issues that may be mitigated by other approaches include the presence of  local minima and saddle points, or the high computational cost of running iterative methods on NISQ hardware.
The Feedback-based Algorithm for Quantum Optimization (FALQON) \cite{MagannFALQON} is an alternative proposed to address these problems.
The main idea is to represent the cost function as a Lyapunov function defined on a quantum system encoding the decision variables, and to use a control law to steer the state to the minimum of this Lyapunov function.
Other works on the topic of combinatorial optimization involve the control of systems to excited states \cite{rahman2024feedbackbasedquantumalgorithmexcited, rahman2024weightedfeedbackbasedquantumalgorithm}, approaches to solve constrained optimization problems \cite{rahman2024} or different ans\"atze for the circuits \cite{Hadfield_2019, Wiersema_2020}.
\\

Furthermore, when run on real hardware, quantum algorithms are affected by errors due to noise. 
One can distinguish between coherent and decoherent errors, depending on whether they disrupt coherence or not. 
In this paper, we focus on coherent control errors, which are modeled as multiplicative perturbations on Hamiltonian evolutions. Previous works address coherent control errors using Lipschitz bounds to find an upper bound on the fidelity of the noisy algorithm output \cite{Berberich:2023zau,funcke2024robustnessoptimalquantumannealing}.

%\subsection{Contributions}
In this paper, we provide a theoretical analysis of the robustness of FALQON, finding that the asymptotic convergence guarantees for the nominal case still hold in presence of systematic coherent control errors. In our theoretical results, we relax the assumptions needed for FALQON to solve combinatorial optimization problems with non-unique optimal solutions. 
Additionally, we derive a robustness bound for independent coherent control errors which leads to the design of a robustified FALQON based on a regularized Lyapunov function. 
The theoretical results are accompanied by simulations ran on a classical computer. % using quantum computing frameworks such as Pennylane and IBM's Qiskit. 

%%%%%%%%%%%%%%%%%%%%%%%%%%%%%%%%%%%%%%%%%%%%%%%%%%%%%%%%%%%%%%%%%%%%%%%%%%%%%%%%
\section{Preliminaries}
\subsection{Combinatorial Optimization setup}
Combinatorial optimization is a class of optimization problems where the set of feasible solutions is discrete.
In this paper, we consider problems of the form
 \begin{equation}
    \min_{x \in {\{0,1\}}^n} l(x), \label{discropt}
 \end{equation}
 where $l(x)$ is the cost function and the decision variable $x$ represents the bit string corresponding to elements in the set. 

In this work, we  use the standard quantum computing notation outlined in \cite{bible}.
In order to tackle problem \eqref{discropt} through quantum computing techniques, the first step is to encode the cost function in a problem Hamiltonian defined as follows: 
\begin{definition}
    Given a decision variable $x \in {\{0,1\}}^n$, a cost function $l:{\{0,1\}}^n \mapsto \mathbb{R}$, and a qubit string $\ket{x} \in (\mathbb{C}^{2})^{\otimes n}$  a problem Hamiltonian $H_{\mathrm{p}} \in \mathbb{H}^{2^n \times 2^n}$,  is a matrix such that
    \begin{equation}
        l(x)\ket{x}=H_{\mathrm{p}} \ket{x}
    \end{equation} 
    for each $x \in {\{0,1\}}^n$.  
\end{definition}
The set $\mathbb{H}^{2^n \times 2^n}$ is defined as $\mathbb{H}^{2^n \times 2^n} :=\{H \in \mathbb{C}^{2^n \times 2^n}: H^\dagger=H\}$.

Problem Hamiltonians are diagonal in the computational basis. Ways to build the problem Hamiltonian with polynomial complexity are described in \cite{Lucas_2014} and \cite{Hadfield_2021}. 
It is immediate to see that, under this formulation, the optimal solution string for \eqref{discropt} is the eigenvector associated to the lowest eigenvalue of $H_{\mathrm{p}}$, hereafter referred to as ground state of $H_{\mathrm{p}}$. %\textit{Should I show it?}
Solving the optimization problem \eqref{discropt} is therefore equivalent to solving 
\begin{equation}
    \min_{\ket{x} \in (\mathbb{C}^{2})^{\otimes n}} \bra{x} H_{\mathrm{p}} \ket{x} \label{eqproblem}.
\end{equation}
\subsection{FALQON}
FALQON was proposed in \cite{MagannFALQON} to solve combinatorial optimization problems without the complicacies that come with Variational Quantum Algorithms (VQAs), in particular the problem of barren plateaus.
It relies on quantum Lyapunov control (QLC) theory and applies it to the optimization setup.

In the following, we introduce the basic idea of FALQON. We begin by describing the evolution of a quantum system with the Schr\"odinger equation, namely
\begin{equation}
    i\dfrac{d}{dt} \ket{\psi(t)}=(H_{\mathrm{p}} + \beta(t)H_{\mathrm{d}})\ket{\psi(t)}, \label{Schrodinger}
\end{equation}
where $H_{\mathrm{p}} \in \mathbb{H}^{2^n \times 2^n}$ is the problem Hamiltonian, $H_{\mathrm{d}} \in \mathbb{H}^{2^n \times 2^n}$ such that $[H_{\mathrm{p}},H_{\mathrm{d}}]\ne 0$ is the driver Hamiltonian,
and $\beta(t)\in \mathbb{R}$ is the input signal. The Hamiltonians have been normalized such that Planck's constant $\hbar=1$. 

Using the cost function in \eqref{eqproblem} as a Lyapunov function $V: (\mathbb{C}^{2})^{\otimes n} \mapsto \mathbb{R}$, and replacing the decision variable $\ket{x}$ by the state $\ket{\psi(t)}$, we get
\begin{equation}
    V(\ket{\psi(t)})= \bra{\psi(t)} H_{\mathrm{p}} \ket{\psi(t)}. \label{lyapunov_function}
\end{equation}

The core idea of FALQON is to choose an input $\beta(t)$ such that $\dfrac{d}{dt}V(\ket{\psi(t)})<0$ to minimize the cost. % We  later show the necessary condition for the system to converge to the global minimum for the cost. 

We begin by taking the derivative of the Lyapunov function with respect to time. 
Using the chain rule we get:
\begin{equation}
    \dfrac{d}{dt}V(\ket{\psi(t)}) = A(t)\beta(t), A(t)=\bra{\psi(t)} i[H_{\mathrm{d}},H_{\mathrm{p}}] \ket{\psi(t)}.
\end{equation}
In order to render the derivative negative we can pick $\beta(t)$ in any way that satisfies the condition 
\begin{equation}
    \beta(t)=-w f(A(t),t) \label{input_law},
\end{equation}
with $w>0$, $f(0,t)=0$ and $f(A(t),t)A(t)>0$, for each $A(t)\ne 0$.
The resulting dynamic, updated with the feedback law \eqref{input_law}, results in
\begin{equation}
    i\dfrac{d}{dt} \ket{\psi(t)}=(H_{\mathrm{p}} - w f(A(t))H_{\mathrm{d}})\ket{\psi(t)}. \label{feedbackdynamics}
\end{equation}
In order to implement the algorithm, one constructs a circuit by simulating the discrete time evolution via a Trotter decomposition \cite[Theorem 4.3]{bible}. Consider the unitary gates
\begin{equation}
    U_\mathrm{p}=e^{-i H_{\mathrm{p}} \Delta t}
\end{equation}
and
\begin{equation}
    U_\mathrm{d}(\beta_t)=e^{-i\beta_t H_{\mathrm{d}} \Delta t}.
\end{equation}

We can approximate the discrete-time evolution of the continuous-time system \eqref{feedbackdynamics} via
\begin{equation}
    \ket{\psi_{t+1}}=U_\mathrm{d}(\beta_t) U_\mathrm{p} \ket{\psi_t} \label{dtschroedinger}, \beta_t=-w f(A_t), 
\end{equation}
with $A_t=\bra{\psi_t} i[H_{\mathrm{d}},H_{\mathrm{p}}] \ket{\psi_t}$.

The circuit is repeatedly run and measured, and the bit string corresponding to the measured state with the lowest energy is selected as the optimal string for the combinatorial optimization problem.
\begin{algorithm}
    \caption{FALQON}
    \label{FALQON_PSEUDOCODE}
    \begin{algorithmic}
        \State 1. \textbf{initialize} $H_{\mathrm{p}}$, $H_{\mathrm{d}}$, $L$, $\Delta t$
        \State 2. initialize the first input \\ \hspace{1cm} $\beta_1 \gets 0$
        \State 3. $t \gets 1$
        \State 4. initialize the state \\ \hspace{1cm} $\ket{\psi_0} \gets \frac{1}{\sqrt{2^n}}\sum_{x} \ket{x}$
        \State 5. implement first layer \\ \hspace{1cm} $\ket{\psi}_1 \gets U_\mathrm{d}(\beta_1)U_\mathrm{p} \ket{\psi_0}$
        \State 6. determine $A_1$ by repeatedly running the circuit \eqref{dtschroedinger}\\ \hspace{1cm} $A_1 \gets \bra{\psi_1} i[H_{\mathrm{d}},H_{\mathrm{p}}] \ket{\psi_1}$
        \State 7. $\beta_2 \gets -A_1$
        \While{$t<L$}
        \State 8. $t \gets t+1$
        \State 9. initialize the state \\ \hspace{1cm} $\ket{\psi_0} \gets \frac{1}{\sqrt{2^n}}\sum_{x} \ket{x}$
        \State 10. implement circuit \\ \hspace{1cm} $\ket{\psi}_t \gets (\prod_{\tau=1}^{t}U_\mathrm{d}(\beta_\tau)U_\mathrm{p}) \ket{\psi_0}$
        \State 11. determine $A_t$ by repeatedly running the circuit \eqref{dtschroedinger}\\ \hspace{1cm} $A_{t} \gets \bra{\psi_{t}} i[H_{\mathrm{d}},H_{\mathrm{p}}] \ket{\psi_{t}}$
        \State 12. $\beta_{t+1} \gets -A_{t}$
        \EndWhile
        \State return $\{\beta_\tau\}_{\tau=1}^L$
    \end{algorithmic}
\end{algorithm}

Various extensions of this algorithm are discussed in \cite{MagannFALQON}, including the multi-input case, the addition of reference perturbation to the control input and the use of iterative schemes to improve the input selection. 
The work also suggests that the algorithm is not only useful on its own but can also be used as an alternative way to initialize a QAOA.

\subsection{Errors and Robustness Measures}
In the following, we  introduce the noise sources we  deal with in this work and the metrics that  be used to study the robustness of FALQON. 
First, we define coherent control errors. 

A coherent control error acting on a quantum gate $\hat{U}=e^{-iH}$ causes an error in the form 
\begin{equation}
    U(\varepsilon)=e^{-(1+\varepsilon) i H}.
\end{equation}

This kind of error could model uncertainties in the control, such as, for example, some miscalibration in the time duration of the input pulse \cite{Berberich:2023zau}.  
The metric we use to measure robustness is the fidelity.
If we now consider an ideal quantum circuit $\hat{\mathbf{U}}=\hat{U}_1\cdots\hat{U}_l,$ with  $\hat{U}_1,\cdots,\hat{U}_l \in \mathbb{U}^n$, and a noise vector $\varepsilon \in \mathbb{R}^l$,
we can define a perturbed quantum circuit $\mathbf{U}(\varepsilon)=U_1(\varepsilon_1)\cdots U_l(\varepsilon_l)$.
The fidelity between the outputs of the two circuits can be defined as follows:
\begin{definition}
    \label{robound}
    Given an ideal output $\hat{\ket{\psi}}=\hat{\mathbf{U}}\ket{\psi_0}$ and a perturbed output ${\ket{\psi(\varepsilon)}}={\mathbf{U}(\varepsilon)}\ket{\psi_0}$ 
    the fidelity $\mathcal{F}(\varepsilon), \mathcal{F}:\mathbb{R}^l \mapsto [0,1]$ is defined as:
    \begin{equation}
        \mathcal{F}(\varepsilon)=|\braket{\hat\psi | \psi(\varepsilon)}|.
    \end{equation} 
\end{definition}
Since quantum states are unit vectors the maximum fidelity is $\mathcal{F}=|\braket{\hat\psi | \hat\psi}|=1$, meaning the noisy state $\ket{\psi(\varepsilon)}$ is equal to the ideal state $\ket{\hat{\psi}}$ up to a global phase difference.

In this work, we  employ this robustness measure to better understand the performance of FALQON in presence of coherent control errors. 

Furthermore, measuring the output of a quantum circuit requires running it multiple times and taking a sample average of the measured outputs. 
In order to describe the different possible behaviors coherent control errors could have in this setup, we  distinguish between systematic and independent coherent control errors. 
We call a coherent control error sequence systematic if it remains constant over each run of the circuit. 
Referring to the notation used in step 10 of algorithm \ref{FALQON_PSEUDOCODE}, given a sequence $\{\varepsilon^t\}_{t=1}^L$, with $\varepsilon^t: \{1,\dots, t\} \rightarrow \mathbb{R}$, the error sequence is  systematic if $\varepsilon^{t_1}(\tau)=\varepsilon^{t_2}(\tau)$ for each $t_1,t_2 \in \{1,\dots,L\}$ and for each $\tau \in \{1 ,\dots, \min\{t_1, t_2\}\}$.
We call the sequence independent if it is not systematic, implying a different error signal in each iteration $t$. 
This distinction is of particular importance for FALQON: systematic errors can be compensated by the feedback structure of the algorithm while, in case of independent error, the input signal is affected by some time-varying error signal. 
We  elaborate the implications of this distinction further in the next section.

\section{ROBUSTNESS OF FALQON}
This section develops an analysis of FALQON's robustness properties under coherent control errors. The closed-loop FALQON dynamics under the effect of such an error can be stated as: 

\begin{equation}
    i\dfrac{d}{dt} \ket{\psi(t)}=(1+\varepsilon(t))(H_{\mathrm{p}} - w f(A(t))H_{\mathrm{d}})\ket{\psi(t)} \label{feedbackCCE}.
\end{equation}

In subsection A, we study the asymptotic convergence of FALQON under systematic coherent control errors.
Then, in subsection B, we derive a robustness bound for it with respect to independent coherent control errors. 
Finally, in subsection C, we take advantage of our results to design a robustified input law to tackle such errors. %\textit{ we?}
Subsections A and C address the problems in the continuous-time QLC setup, while subsection B addresses the discrete time implementation of FALQON directly. 
This is done for ease of presentation. In particular, the results from subsections A and C can be readily transferred to discrete time, and results analogous to subsection B can be shown in continuous time following \cite{funcke2024robustnessoptimalquantumannealing}.
\subsection{Asymptotic Convergence for Systematic Errors}
In the following, we  study the convergence properties of the noisy continous-time dynamics of the FALQON algorithm, namely  \eqref{feedbackCCE}.
First, we  outline the hypotheses needed for convergence in the noiseless case. We  write $p_i$ and $q_i$ for the i-th eigenvalue and eigenvector of $H_{\mathrm{p}}$, respectively, written in increasing order ($p_0$ being the lowest). 
\begin{assumption}
    Consider the following assumptions: 
    \begin{itemize}
        \item[(a)] $H_{\mathrm{p}}$ has no degenerate eigenvalues, i.e., $p_i \ne p_j$ for $i \ne j$.
        \item[(b)]$H_{\mathrm{p}}$ has no degenerate eigenvalue gaps, i.e., $\omega_{ji} \ne \omega_{lk}$ for $(i, j) \ne (k, l)$, where $\omega_{ji} = q_j-q_i$ is the gap between the i-th and j-th eigenvalues of $H_{\mathrm{p}}$.
        \item[(c)]$H_{\mathrm{d}}$ is fully connected, i.e., $\bra{q_i} H_{\mathrm{d}} \ket{q_j} \ne 0$ for $i \ne j$.
        \item[(d)]The initial condition has an energy lower than the energy of the first excited state, i.e. \\$V(\ket{q_0}) < V(\ket{\psi(t=0)}) < V(\ket{q_1})$.        
    \end{itemize}
    \label{A1A4}
\end{assumption} 

\begin{theorem}
    Let $\ket{\psi(t)}$ be the quantum state evolving under the dynamics given in \eqref{feedbackCCE}, and let assumptions \ref{A1A4} (a-d) hold. 
    Assume also $|\varepsilon(t)|<\bar{\varepsilon}$ for all $t>0$, with some $\bar{\varepsilon} < 1$, and assume the error signal $\varepsilon(t)$ to be a systematic error. Then 
    \begin{equation}
        \lim_{t \to \infty} \ket{\psi(t)}=\ket{q_0}.
    \end{equation}
    \label{asymptotic_convergence}
\end{theorem}

\begin{proof}
The proof follows the approach taken for the noiseless case \cite{Grivopoulos} with additional considerations due to the presence of noise.
By differentiating the Lyapunov function $V(\ket{\psi(t)})$ with respect to time we obtain:
\begin{align}
    & \dfrac{d}{dt}V(\ket{\psi(t)}) \nonumber \\
    &=\bra{\psi(t)} i(1+\varepsilon (t))[H_{\mathrm{p}}-wf(A(t))H_{\mathrm{d}},H_{\mathrm{p}}] \ket{\psi(t)} \nonumber \\
    &= \bra{\psi(t)} i(1+\varepsilon (t))[-wf(A(t))H_{\mathrm{d}},H_{\mathrm{p}}] \ket{\psi(t)} \nonumber \\
    &= -w\bra{\psi(t)} i[H_{\mathrm{d}},H_{\mathrm{p}}] \ket{\psi(t)}(1+\varepsilon (t))f(A(t)) \nonumber\\
    &= -w(1+\varepsilon (t))A(t)f(A(t)).
\end{align}
According to \eqref{input_law}, the derivative of the Lyapunov function is $0$ if and only if $A(t)=0$.
Assuming $A(t)=0$, the dynamics in \eqref{feedbackCCE} reduce to an autonomous ordinary differential equation that we can write as 
\begin{equation}
    i \dfrac{d}{dt}\ket{\psi(t)}= (1+\varepsilon(t))H_{\mathrm{p}} \ket{\psi(t)}. \label{autonomousdynamics}
\end{equation}
The (forward invariant) solution trajectory for \eqref{autonomousdynamics}, $\ket{\bar\psi(t)}$, can be expressed as 
\begin{equation}
    \ket{\bar{\psi}(t)}=e^{-i\eta(t) H_{p}} \ket{\bar{\psi}(0)}=\sum_{i=1}^{N} c_i e^{-i\eta(t) p_i} \ket{q_i} \label{fitraj},
\end{equation}
where $\eta(t)=\int_{0}^{t} (1+\varepsilon (\tau)) d\tau$, and $c_i$ is the complex coefficient associated to the eigenvector $q_i$ representing $\ket{\bar\psi(t)}$ as a linear combination of the computational basis.
Furthermore, since the input is chosen as per \eqref{input_law}, requiring $\beta(t)=0$ is equivalent to imposing
\begin{equation}
    \bra{\bar{\psi}(t)}[H_{\mathrm{p}}, H_{\mathrm{d}}]\ket{\bar{\psi}(t)}=0 \label{fiinput}.
\end{equation}
We can now substitute \eqref{fitraj} in \eqref{fiinput} and obtain
\begin{equation}
    -i(1+\varepsilon(t))\sum_{i,j=1}^{N} (p_i - p_j) c_i c_j^* e^{-i  \omega_{i,j} \eta(t)} \bra{q_j} H_{\mathrm{d}} \ket{q_i} = 0, \label{griveps}
\end{equation}
for all $t>0$.

The exponential trajectories $e^{-i  \omega_{i,j} \eta(t)}$ are always linearly independent for each $(i,j)$ pair, due to assumption 1 (b) and $\eta(t)$ being the same for all the pairs. To guarantee this we also impose $\bar{\varepsilon} < 1$, so that $\eta(t)>0$.
Assumption 1 (b) and (c) ensure that none of the terms in the sum on the left hand side of \eqref{griveps} are identically zero independently from the trajectory $\ket{\bar{\psi}(t)}$. 
Imposing these conditions, a necessary and sufficient condition for \eqref{griveps} to be identically $0$ is $c_i c^*_j=0, \forall i,j < N$.
In order to satisfy this condition only one coefficient can be non-zero, namely $\exists! i < N | c_i \ne 0$.
The last statement is true only for eigenvectors of $H_{\mathrm{p}}$. 
Furthermore, since the cost function \eqref{lyapunov_function} is non-increasing and due assumption 1 (d), the system converges to the ground state.
\end{proof}
The result proves that in the noisy case we have the same convergence guarantees as in the noise-free case. Intuitively, since we are assuming a systematic multiplicative error, the feedback structure of the algorithm is able to mitigate it and cancel it as the nominal input goes to zero.

In view of the proof, we can see that assumption 1 (a) is needed to prove convergence to a specific eigenvector. 
In practical examples, however, it can happen that multiple different eigenvectors lead to the same cost. 
If this is the case, we can prove convergence for the algorithm to the 
global minimum of the cost, which is sufficient for us to solve the combinatorial optimization problem. 
We state and prove the theorem for the error-affected dynamics, but it is trivial to extend it to the noiseless case as it can be treated as a special case in which $\bar{\varepsilon} = 0$.

\begin{theorem}
    Let $\ket{\psi(t)}$ be a quantum state evolving under the dynamics given in \eqref{dtschroedinger}, and let assumptions \ref{A1A4}, (b-d) hold. 
    Assume also $\bar{\varepsilon} < 1$. Then 
    \begin{equation}
        \lim_{t \to \infty} \bra{\psi(t)} H_\mathrm{p}\ket{\psi(t)}=p_0.
    \end{equation}
\end{theorem}

\begin{proof}
    Consider the equality \eqref{griveps}. For each pair $i,j$ in the summation, in order for it to be 0, at least one of the following must be true:
    \begin{subequations}
    \label{cond:conds}
    \begin{align}
        & p_i-p_j=0, \label{cond:1} \\% \text{ if } p_i=p_j=p^{\star}, \\
        & c_i^\star c_j =0. \label{cond:2}% \text{ if } p_i \ne p^{\star} \text{ or } p_j \ne p^{\star}. 
    \end{align}
    \end{subequations}
    If assumption 1 (a) is false, there is an eigenvalue $p^\star$ such that at least two eigenvectors are associated to it.
    For all eigenvector pairs where both eigenvectors are associated to $p^\star$, condition \eqref{cond:1} holds, so \eqref{cond:2} is not necessary. Contrarily, for all eigenvectors pairs where at least one is not associated to $p^\star$, condition \eqref{cond:2} must hold. 
    
    This means that all the non-zero components of $\ket{\psi(t)}$ are eigenvectors associated to the eigenvalue $p^\star$.
    If assumption 1 (d) holds, and since we proved that the cost \eqref{lyapunov_function} is non-increasing, the only possible value for $p^\star$ is $p^\star=p_0$.
    Since $\ket{\psi(t)}$ is a unit vector and all its components are associated to $p_0$, this proves the result.  
\end{proof}

We stated and proved the theorem for the error-affected dynamics, but it is trivial to extend it to the noiseless case as it can be treated as a special case in which $\bar{\varepsilon} = 0$. 
This also relaxes the required assumptions for FALQON outlined in \cite{MagannFALQON}.

\subsection{Robustness Bound for Independent Errors}
Theorem 1 and 2 only hold if we assume systematic noise from the system. It can be shown that, if we consider independent noise, the feedback argument used to justify convergence for FALQON under systematic errors no longer applies.
However, we can find bounds for the impact of independent coherent control errors on the final state. 
In the following section, we  find an upper bound on the error on the final state and use it to design a robustified input law.

We  now examine \eqref{feedbackCCE}, assuming an independent coherent control error signal, in order to find a fidelity lower bound. 
Due to the multiple executions required to build each FALQON layer, it may not always be possible to assume fixed error trajectories. Indeed, the work \cite{9951287} shows that time varying errors are a relevant problem on currently available quantum hardware. 
Our main result is the following lemma: 
\begin{lemma}
    Let $\ket{\hat{\psi}}$ be a noise-free output for a FALQON circuit with depth $l$. 
    Given a scalar $\bar{\varepsilon}>0$ and a noise vector $\varepsilon \in \mathbb{R}^l$ such that $||\varepsilon||_\infty < \bar{\varepsilon}$, $\varepsilon$ representing an independent coherent control error sequence, let $\ket{\psi(\varepsilon)}$ be the output of the circuit affected by noise.
    The fidelity can be bounded as 
    \begin{equation}
        |\braket{\psi(\varepsilon)| \hat\psi}| \ge 1-\frac{L_{\mathrm{FALQON}}^2\bar\varepsilon^2}{2},
    \end{equation}
    
    where 
    \begin{equation}
        L_{\mathrm{FALQON}}= \sum_{t=1}^{l} \Delta t ||H_{\mathrm{p}} + \beta_t H_{\mathrm{d}}||_2 \label{falqbound},
    \end{equation}
\end{lemma}

\begin{proof}

Berberich, Fink and Holm \cite[Th. 2.1]{Berberich:2023zau} proved that if we are able to find a Lipschitz bound $L$ such that
\begin{equation}
    ||\ket{\hat\psi} - \ket{\psi(\varepsilon)}||_2 \le L ||\varepsilon||_\infty\text{, }\forall \varepsilon \in \mathbb{R}^l,
\end{equation}
then the fidelity can be bounded as 
\begin{equation}
    |\braket{\psi(\varepsilon)| \hat\psi}| \ge 1-\frac{L^2\bar\varepsilon^2}{2}. \label{berberbound}
\end{equation}
Therefore we need to find a Lipschitz bound for FALQON in order to bound fidelity. 
The work \cite{funcke2024robustnessoptimalquantumannealing} determines a Lipschitz bound for the error assuming continous time evolution of a quantum system with time-varying Hamiltonian $H(t)$ affected by coherent control error.
In particular, \cite{funcke2024robustnessoptimalquantumannealing} shows that 
\begin{equation}
    ||\ket{\hat\psi}-\ket{\psi(\varepsilon)}||_2 \le \int_{0}^{T} ||H(\tau)||_2 d\tau ||\varepsilon||_{\infty}.
\end{equation}
Recalling the system dynamics \eqref{Schrodinger} and the input law \eqref{input_law} we get:
\begin{align}
    \nonumber & L_{\mathrm{FALQON}} = \int_{0}^{T} ||H(\tau)||_2 d\tau \\
    \nonumber & = \int_{0}^{T} ||H_{\mathrm{p}} +\beta(\tau) H_{\mathrm{d}}||_2 d\tau \\
    \nonumber & \le  \sum_{t=1}^{l} \int_{(t-1)\Delta t}^{t\Delta t} ||H_{\mathrm{p}} + \beta_t H_{\mathrm{d}}||_2 d\tau  \\
    & =  \sum_{t=1}^{l} \Delta t ||H_{\mathrm{p}} + \beta_t H_{\mathrm{d}}||_2, 
\end{align} 
with $\beta_t=A_{t-1}=A((t-1)\Delta t)$.

By inserting $L_{\mathrm{FALQON}}$ into \eqref{berberbound} we prove the lemma. 
\end{proof}

\subsection{Robust FALQON algorithm}
Having found a bound for the fidelity in Lemma 1, we are able to design a robustified input law for FALQON. 
In order to improve robustness, we introduce a penalty term in the cost function that increases as the upper bound increases. 
Our penalty term is chosen to be proportional to $\beta(t)^2$ because this penalizes an upper bound on \eqref{falqbound}, namely
\begin{align}
    & L_{\mathrm{FALQON}} = \sum_{t=1}^{l} \Delta t ||H_{\mathrm{p}} + \beta_t H_{\mathrm{d}}||_2 \nonumber \\
    & \le \sum_{t=1}^{l} \Delta t ||H_{\mathrm{p}}||_2 + ||\beta_t H_{\mathrm{d}}||_2.
\end{align}
We compute the input with an updated Lyapunov function
\begin{equation}
    V_\lambda(\ket{\psi(t)}, \beta(\cdot))=V(\ket{\psi(t)}) + \lambda \int_{0}^{t} \beta^2(\tau) d\tau.
\end{equation} 
Note that this approach is analogous to the regularization proposed for VQAs in \cite{Berberich:2023zau}.

The derivative of the updated Lyapunov function is 
\begin{align}
   &  \dfrac{d}{d_t}V_\lambda(\ket{\psi(t), \beta(\cdot)})  \nonumber \\ 
   & =  \bra{\psi(t)} i[H_{\mathrm{p}}+\beta(t)H_{\mathrm{d}},H_{\mathrm{p}}] \ket{\psi(t)} + \lambda \beta^2(t)   \nonumber \\ 
   & =  \beta(t)\bra{\psi(t)} i[H_{\mathrm{d}},H_{\mathrm{p}}] \ket{\psi(t)} + \lambda \beta^2(t) \label{robder}.
\end{align}
Since \eqref{robder} has to be negative to make the cost decrease we can rewrite it using the definition of $A(t)$ and solve the inequality
\begin{equation}
    \lambda  \beta^2(t) + A(t)\beta(t) < 0.
\end{equation}
The interval for $\beta(t)$ in which \eqref{robder} takes negative values is 
\begin{equation}
    \beta(t) \in \begin{cases}
        [\frac{-A(t)}{\lambda}, 0] \text{ if } A(t)>0, \\
        [0, \frac{-A(t)}{\lambda}] \text{ else}
    \end{cases} \label{interval}.
\end{equation}
Furthermore, by taking the derivative of \eqref{robder} with respect to $\beta$ we can easily verify that a minimum for \eqref{robder} can be found in 
\begin{equation}
    A(t)+2\lambda\beta(t)=0 \Leftrightarrow \beta (t)=-\frac{A(t)}{2\lambda},
\end{equation}
which is exactly the center of interval \eqref{interval}.

This feedback law for the parameter $\beta$ fits the requirements stated in \eqref{input_law}, since it is exactly the same choice for $f(A(t),t)$ used in their experiments ($f(A(t),t)=A(t)$), with a suitable choice of $\lambda$. According to our derivation, the regularization term influences the input choice as a weight $w=\frac{1}{2\lambda}$, meaning that increasing regularization decreases the norm of $\beta$ and, thereby, improves robustness.
Noticeably, the choice of larger values for $\lambda$ admits a trade-off between robustness and optimality, leading to a slower convergence due to the reduced gain.
It is immediate to prove that for $\lambda= \frac{1}{2}$, $\beta(t)=-A(t)$, which is exactly the input law chosen in \cite{MagannFALQON}.

\section{NUMERICAL RESULTS}
In this section, we  test FALQON on a MaxCut problem. MaxCut aims at finding a partition of a graph that maximizes the number of edges cut. 
We  test the algorithm on an 8-node random regular graph, built according to the Erdos-Renyi model. The graph is unweighted. Furthermore, the MaxCut problem on the graph we use has non-unique optimal solutions.

The discretization step is selected empirically. We choose $\Delta t=0.05s$ to achieve good performance. 
Coherent control errors are sampled uniformly from the interval $[-\bar{\varepsilon},+\bar{\varepsilon}]$.
The parameter $\bar{\varepsilon}$ varies among experiments. 

In Section IV.A, we show the robustness properties of FALQON with respect to systematic coherent control errors.
In Section IV.B, we analyze the robustness of FALQON with respect to independent coherent control errors and compare it to robust FALQON. 
All the simulations are developed using Pennylane \cite{bergholm2022pennylaneautomaticdifferentiationhybrid}. The code is available at the following URL: \url{https://github.com/MirkoLegnini/Robust_FALQON}. 
\subsection{Systematic coherent control errors}
The purpose of this experiment is to show that systematic coherent control errors do not affect the asymptotic convergence properties of FALQON. 
We run the experiment with a maximum depth for the circuit $l=1000$. 
We repeat the experiment for increasingly strong error levels, from $\bar{\varepsilon}=0.1$ up to $\bar{\varepsilon}=0.9$.  
Figure \ref{fig1} shows the error on the cost over iterations. 
\begin{figure}    
    \includegraphics[width=0.45\textwidth]{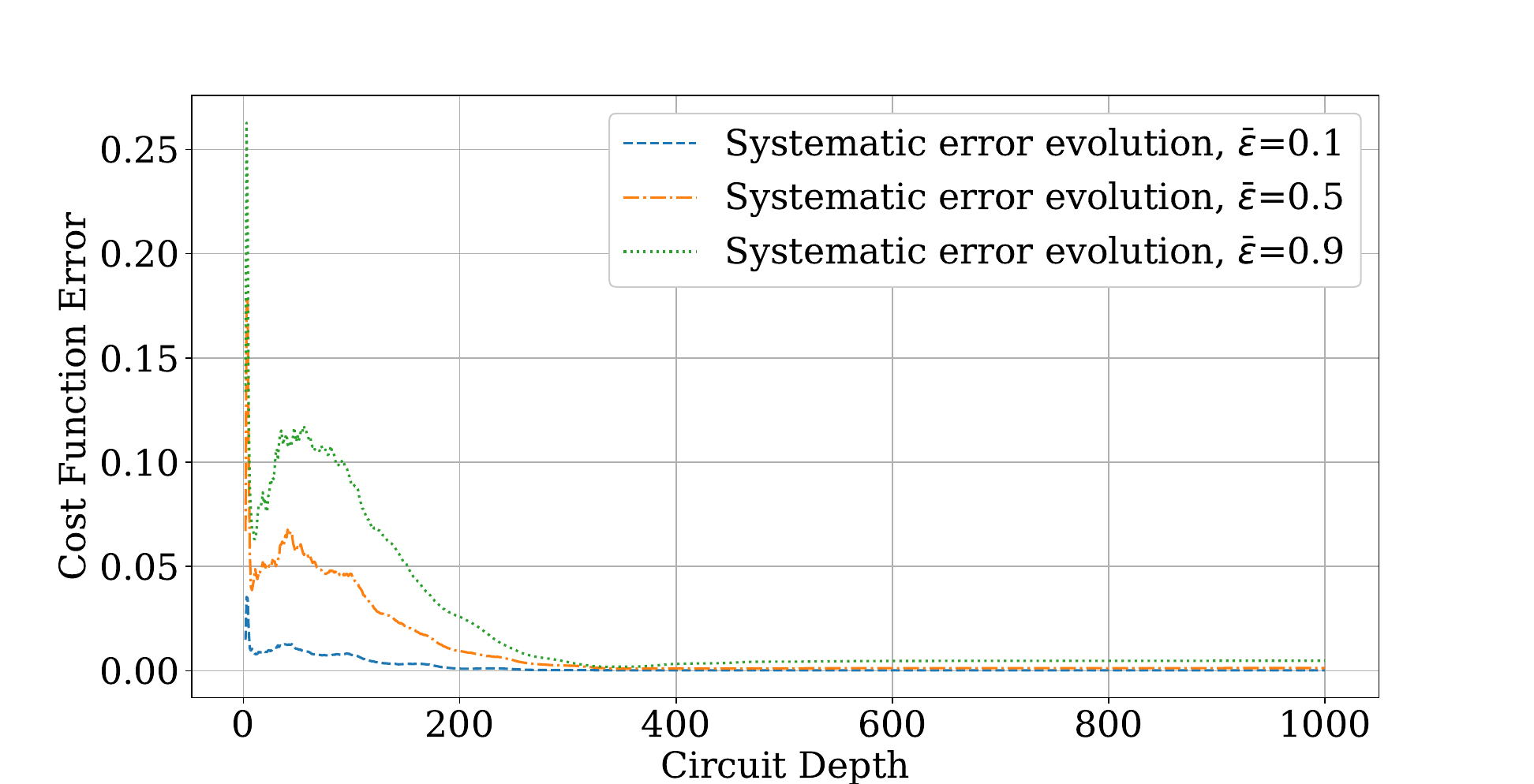}
    \caption{\label{fig1}Evolution of error on cost over layers under systematic CCE}
\end{figure}
The results show that, regardless of the entity of the noise, FALQON is asymptotically robust to systematic coherent control errors.
There is, however, a small difference in the behavior during the transients.

\subsection{Independent coherent control errors}
Here, we investigate the robustness of FALQON with respect to independent coherent control errors. 
We run the experiment with a maximum depth for the circuit $l=200$. 
We run the experiment 50 times with different noise samples and compute the sample standard deviation.
We repeat the experiment for robust FALQON with $\lambda=1.0$ and for the standard FALQON (equivalent to $\lambda=0.5$).
\begin{figure}
    \begin{subfigure}[h]{0.5\textwidth}
        \includegraphics[width=\textwidth]{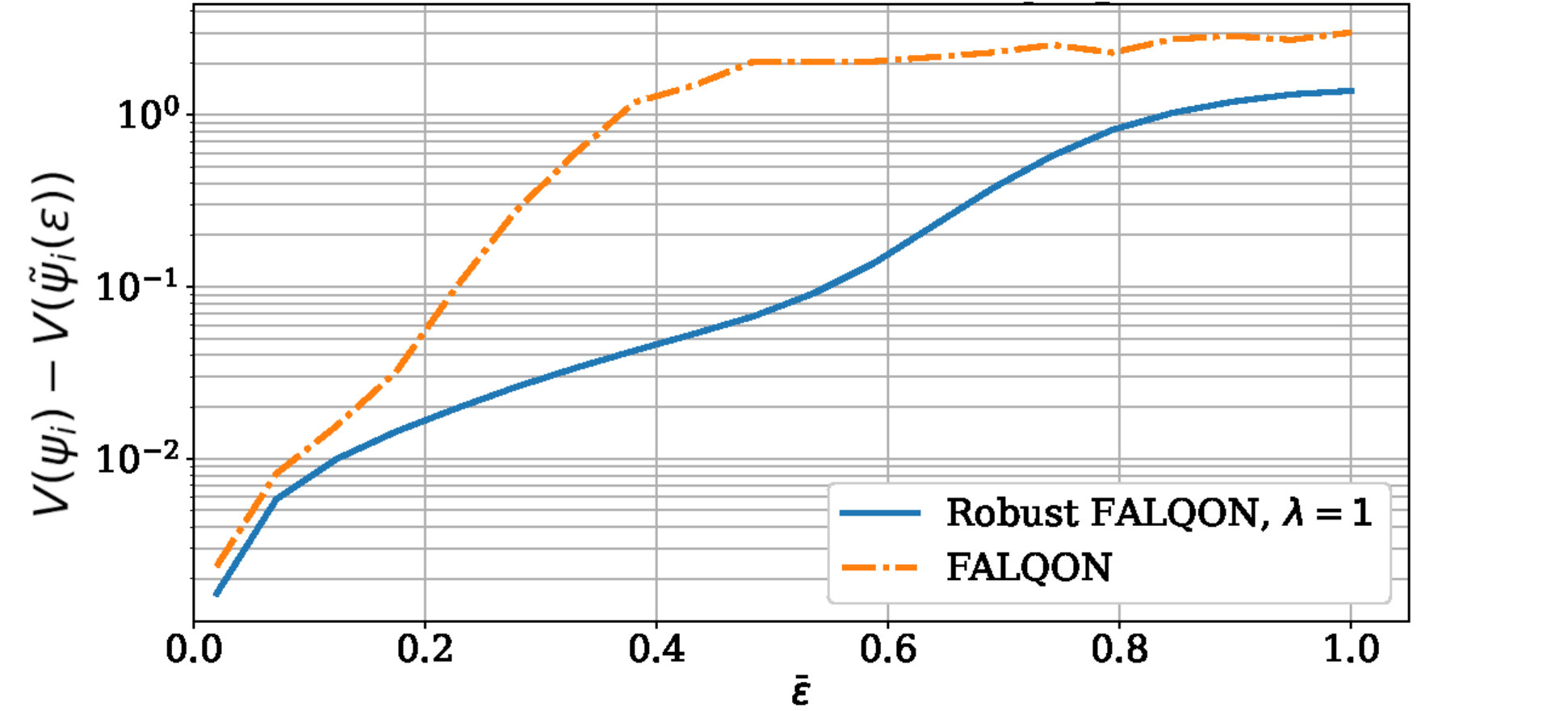}
        \caption{\label{fig:image2} Errors on final layer for varying $\bar\varepsilon$}    
    \end{subfigure}
\\

    \begin{subfigure}[h]{0.5\textwidth}
        \includegraphics[width=\textwidth]{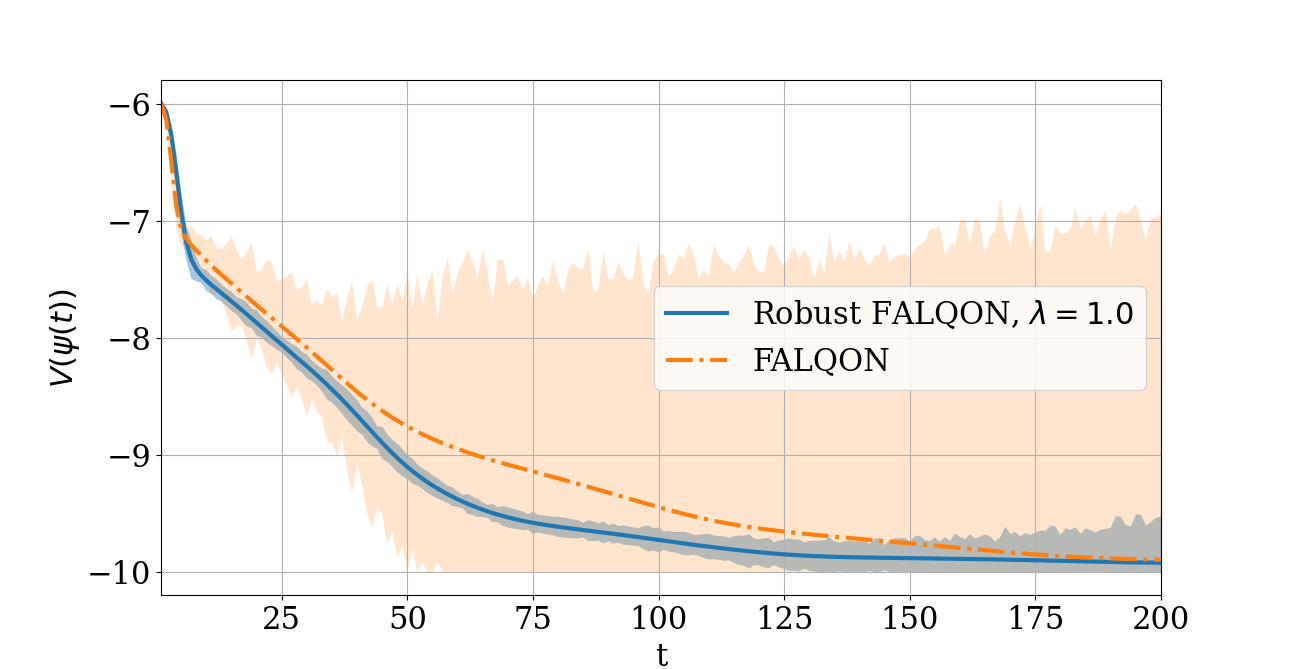}
        \caption{\label{fig:image1}Evolution of cost over layers, $\bar\varepsilon=0.25$ }    
    \end{subfigure}

    \caption{Errors on the final layer for varying $\bar\varepsilon$ (Figure \ref{fig:image2}) and evolution of cost over layers under independent CCE (Figure \ref{fig:image1})}
\end{figure}

The numerical results show that noise rejection improves with a suitable choice of $\lambda$.
In particular, Figure \ref{fig:image2} shows that robust FALQON has considerably improved performance.
Figure \ref{fig:image1} also suggests that a robust FALQON can also accelerate the convergence of the algorithm.

\section{CONCLUSIONS AND FUTURE WORKS}

In this work, we investigated the robustness properties of FALQON \cite{MagannFALQON}.
In the first part, we examined the impact of coherent control errors on the algorithm's convergence. We found that systematic coherent control errors do not affect asymptotic stability for the algorithm.
We  then provided an upper bound for the fidelity of the final state in the presence of independent coherent control errors. 
Starting from this upper bound, we designed a robust version of FALQON. Taking inspiration from classical optimization, the proposed idea was to introduce a penalty term related to the norm of the learned parameters. 
We provided numerical evidence for our theoretical results. 
The numerical results demonstrate that, for a suitable parameter choice, the robust FALQON admits superior performance in the presence of noise.
\\
Possible further developments involve an analysis of the impact of the discretization step $\Delta t$ on of the algorithm, which could lead to an optimal discretization step selection to improve robustness and convergence.

\bibliographystyle{ieeetr}
\bibliography{bibliography}
\end{document}